\documentclass[12pt]{article}

\usepackage{enumerate}

\usepackage{amsmath, amsthm, amssymb, bm, natbib}
\usepackage{threeparttable,booktabs,lscape,pdflscape,multicol}
\usepackage{multirow,colortbl,rotating}

\newcommand{\eps}{\epsilon}
\newcommand{\Xp}{X^p}

\newcommand{\hbetasp}{\widehat \beta_{S-p}}
\newcommand{\Sigp}{\Sigma^p}
\newcommand{\rhop}{\rho^p}

\newcommand{\tsigp}{\widetilde \Sigma^p}
\newcommand{\trhop}{\widetilde \rho^p}

\newcommand{\hSig}{\widehat \Sigma}

\newcommand{\trho}{\tilde \rho}

\newcommand{\Dsp}{D_{SS}^p}

\newcommand{\linf}[1]{\| #1 \|_{\infty}}

\def\bzero{{\bf 0}}
\def\bone{{\bf 1}}
\def\by{{\bf y}}

\def\bB{{\bf B}}

\def\bD{{\bf D}}

\def\bI{{\bf I}}

\def\bK{{\bf K}}

\def\bU{{\bf U}}

\def\bX{{\bf X}}

\def\bZ{{\bf Z}}

\def\thick#1{\hbox{\rlap{$#1$}\kern0.25pt\rlap{$#1$}\kern0.25pt$#1$}}

\def\bbeta{\boldsymbol{\beta}}

\def\bepsilon{\boldsymbol{\epsilon}}

\def\btheta{\boldsymbol{\theta}}

\def\brho{\boldsymbol{\rho}}

\def\bSigma{\boldsymbol{\Sigma}}

\def\smbalpha{\boldsymbol{{\scriptstyle{\alpha}}}}

\def\smbalpha{\widehat{\smbalpha}}

\def\hbar{\bar{ h}}

\def\mybox#1{\vskip1mm \begin{center}
        \hspace{.0\textwidth}\vbox{\hrule\hbox{\vrule\kern6pt
\parbox{.9\textwidth}{\kern6pt#1\vskip6pt}\kern6pt\vrule}\hrule}
        \end{center} \vskip-5mm}
\def\lboxit#1{\vbox{\hrule\hbox{\vrule\kern6pt
      \vbox{\kern6pt#1\vskip6pt}\kern6pt\vrule}\hrule}}

\def\thickboxit#1{\vbox{{\hrule height 1mm}\hbox{{\vrule width 1mm}\kern6pt
          \vbox{\kern6pt#1\kern6pt}\kern6pt{\vrule width 1mm}}
               {\hrule height 1mm}}}

\def\fat#1{\hbox{\rlap{$#1$}\kern0.25pt\rlap{$#1$}\kern0.25pt$#1$}}

\def\sgn{\mbox{sgn}}

\usepackage{graphicx,verbatim,epstopdf}
\usepackage{caption}
\usepackage{fancybox,url}
\usepackage{algorithm}
\usepackage{algorithmic}
\usepackage{xcolor}
\def\bSig\mathbf{\Sigma}

\newcommand{\blind}{0}

\theoremstyle{thmstyleone}%
\newtheorem{theorem}{Theorem}
\theoremstyle{thmstyletwo}%
\newtheorem{lemma}{Lemma}%
\newtheorem{con}{Condition}%
\theoremstyle{thmstylethree}%

\begin{document}

\def\spacingset#1{\renewcommand{\baselinestretch}%
{#1}\small\normalsize} \spacingset{1}


\if0\blind
{
  \title{\bf High-dimensional log contrast models with measurement errors}
  \author{Wenxi Tan, Lingzhou Xue  \hspace{.2cm}\\
    Department of Statistics, Penn State University\\
    Songshan Yang\\
    Center for Applied Statistics and Institute of Statistics and Big Data,\\ Renmin University of China\\
    Xiang Zhan \\
    Department of Biostatistics, Beijing International Center \\for Mathematical Research and Center for Statistical Sciences, \\Peking University 
    }
\date{}
  \maketitle
} \fi

\if1\blind
{
  \bigskip
  \bigskip
  \bigskip
  \begin{center}
    {\LARGE\bf High-dimensional log contrast models with measurement errors}
\end{center}
  \medskip
} \fi
\bigskip

\begin{abstract}
High-dimensional compositional data are frequently encountered in many fields of modern scientific research. In regression analysis of compositional data, the presence of covariate measurement errors poses grand challenges for existing statistical error-in-variable regression analysis methods since measurement error in one component of the composition has an impact on others. To simultaneously address the compositional nature and measurement errors in the high-dimensional design matrix of compositional covariates, we propose a new method named \textbf{Er}ror-\textbf{i}n-\textbf{c}omposition (Eric) Lasso for regression analysis of corrupted compositional predictors. Estimation error bounds of Eric Lasso and its asymptotic sign-consistent selection properties are established. We then illustrate the finite sample performance of Eric Lasso using simulation studies and demonstrate its potential usefulness in a real data application example. 
\end{abstract}

\noindent%
\vfill

\newpage
\spacingset{1.9} 

\section{Introduction} \label{sec:intro}
Compositional data, representing relative proportions or percentages of different parts that make up a whole, have a wide range of applications in many fields, including geology, ecology, social sciences and biology. In biological and biomedical research, compositional data primarily arise from  high-throughput sequencing technologies-based profiling experiments,  which share a similar measurement process in which the total abundance information is lost and most sequence counts reflect only the relative abundances (i.e., compositional) information of unique sequences of interest \citep{vandeputte2017}. Regression analysis with these compositional covariates is essential to disentangle the relationships between compositions and an outcome of interest. Due to compositionality, traditional linear regression models fail for regression analysis with compositional predictors. To address the ``curse of compositionality'' in regression analysis, the log contrast model was proposed in the context of experiments with mixtures \citep{aitchison1984}. Since then, multiple extensions have been developed, especially in the high-dimensional settings \citep{hron2012,lin2014,shi2016,wang2017,randolph2018,srinivasan2021,combettes2021,mishra2022,shi2022}.  

Much of the existing work on high-dimensional log contrast regression has focused on the clean data case. However, measurement errors are ubiquitous in many scientific endeavors. Taking the compositional sequence count data in biomedical research that motivate our study as an example, measurement errors may occur at any stage of the experimental workflow, such as DNA extraction, PCR amplification, sequencing process and even bioinformatics preprocessing procedure \citep{mclaren2019}. Moreover, studies have reported that the problem of data contamination may also be related to genome databases with a large number of mis-labeled sequences, which could potentially lead to a wrong sequencing read counts inflated by orders of magnitude \citep{gihawi2023}.
These measurement errors need to be well accommodated in statistical analysis in order to avoid potential misleading or invalid scientific findings \citep{gihawi2023}.

These measurement errors in sequencing studies are often referred to as the sequence bias by many authors working in the field of computational biology \citep{mclaren2019,zhao2021}. In those sequencing experiments, when the research of interest is an individual count variable, the issue of sequencing bias can be appropriately addressed or attenuated by multilevel modeling techniques, such as the Beta-Binomial regression \citep{martin2020} or the Poisson-Gamma model \citep{jiang2023}. However, beyond marginal analysis, it is much more difficult to analyze contaminated compositional predictors normalized from multiple count variables \citep{shi2022} or cell-type proportions associated with uncertainties \citep{cai2022},
since measurement error in one component has a ripple effect on other components due to the compositional constraint. Similar to the regression analysis, we are facing the ``curse of compositionality'' again in measurement error modelling for compositional data.

The canonical model for high-dimensional regression analysis is expressed as $\by=\bX\bbeta^*+\bepsilon$, where $\by=(y_1,\ldots,y_n)^T$ is the response vector, $\bX \in R^{n \times p}$ is the design matrix of high-dimensional covariates and $\bbeta^* \in R^p$ is the regression coefficient vector of interest. In many applications, $\bX$ may not be accurately measured and a corrupted version $\bZ$ of $\bX$ is often available. In the literature of high-dimensional statistics, many versatile methods and theories have been developed for inference of $\bbeta^*$ based on $\bZ$ \citep{rosenbaum2010,loh2012,belloni2017,datta2017}. 
In terms of compositional data, each row of $\bX$ and $\bZ$ belongs to simplex $\mathcal{S}^{p}=\{(x_1,\ldots,x_p): x_j>0, \sum_{j=1}^p x_j=1 \}$. Clearly, measurement error in one component has a ripple effect on other components due to the compositional constraint and thus, existing statistical methods and theories are not directly applicable to measurement error problems of compositional data. 

A more recent paper considered the measurement error problem in the framework of log contrast models for regression analysis with compositional predictors \citep{shi2022}. In particular, the variable correction regularized estimator proposed by \citet{shi2022} requires the knowledge of observed counts and leverages the Direchlet-Multinomial distribution to correct observed counts towards unobservable underlying compositions. However, as pointed out in a recent paper \citep{gihawi2023}, sequence read counts could be inflated by many orders of magnitude due to potential contamination issues of draft reference genomes. Assumptions of the variable correction regularized estimator no longer hold when measurement errors in counts are extremely huge. Moreover, it is sometimes less meaningful to analyze counts from different platforms than compositions \citep{allali2017}, which limits the applicability of variable correction regularized estimator to large cohort studies where samples are typically sequenced at different locations/batches using different platforms. These potential limitations motivate our investigation to take a different approach to fill this research gap.

To develop a new method addressing the aforementioned limitations,  we utilize the recently proposed mathematical model to characterize the issue of sequence bias in next-generation sequencing experiments \citep{mclaren2019,mclaren2022}. To further accommodate compositionality in the design matrix, we build our regression analysis framework upon the Aitchison log contrast model \citep{aitchison1982,aitchison1984} and then propose a new method named textbf{Er}ror-\textbf{i}n-\textbf{c}ompositional (Eric) Lasso to handle corrupted high-dimensional compositional predictors. The error bound of our Eric Lasso estimator along with its selection sign consistency property is established. In summary, we propose a new method to simultaneously address both compositional nature and measurement errors in regressors, which distinguishes it from existing ones. The novelty of our work lies in both a new methodology with desirable statistical properties and interpretations, but also the particular application which is an important and timely problem for which no satisfactory analysis methods exists so far.

The rest of this article is organized as follows. In Section \ref{method}, we first introduce some background on regression analysis with compositional covariates and on existing literature about bias correction for error-in-variable in log contrast regression. Then, we propose our \textbf{Er}ror-\textbf{i}n-\textbf{c}ompositional (Eric) Lasso method to handle regression analysis with contaminated compositional covariates. The estimation error bounds of our method and its asymptotic sign consistency selection properties are established in Section \ref{thm}. We next demonstrate the superior performance of our method both using simulation studies and real data analysis application examples in Section \ref{simu} and Section \ref{rda}, respectively. This article concludes with discussion in Section \ref{dis}. Proofs of theoretical results are provided in online supplementary materials.

\section{Methods} \label{method}
\subsection{Preliminaries}
Let $y_i$ and $(U_{i1},\ldots,U_{ip})$ denotes response of interest and $p$ compositional covariates (i.e., $\sum_{j=1}^p U_{ij}$=1 and $U_{ij}>0, \forall j$) measured from the $i$th sampling unit. To study relationships between response and compositional predictors, the following linear log contrast model \citep{aitchison1984} has been widely used:
\begin{equation} \label{true}
y_i=\sum_{j=1}^p log(U_{ij}) \beta^*_j + \epsilon_i, \:  s.t., \: \sum_{j=1}^p \beta^*_j=0,
\end{equation}
where the intercept term is omitted if both the outcome and predictors are centered. One most remarkable feature in model \eqref{true} is the zero-sum constraint on regression coefficients, which is essential to guarantee some basic principles in compositional data analysis: scale invariance, permutation invariance and subcompositional coherence \citep{lin2014,greenacre2023}. These principles are often necessary for statistically meaningful interpretations of compositional data analysis results \citep{billheimer2001}. To understand this, one can first see that each individual component $j$ of the compositional vector only carries relative information and thus $\beta_j^*$ itself is less meaningful in compositional data analysis. On the other hand, contrast between two coefficients $\beta_j^*-\beta_l^*$ is meaningful in that it can measure the relative importance of component $j$ and $l$. To get rid of potential bias in selecting a specific reference level $l$, one can use the quantity $\frac 1 p\sum_{l=1}^p(\beta^*_j-\beta^*_l)$ to measure the relative importance of component $j$ compared to the remaining ones in the compositional vector. This quantity reduces to $\beta^*_j$ if and only if $\sum_{l=1}^p \beta^*_l=0$ holds. In other words,  coefficient  $\beta^*_j$ can measure the relative importance of the $j$th component of the compositional vector under this zero-sum constraint. Therefore, it is crucial to develop interpretable regression analysis methods for compositional data under this zero-sum constraint in practice \citep{aitchison1982,lin2014,srinivasan2021,shi2022}.

In many scenarios, compositions $U_{ij}$ are not observable or measured with errors. For example, in a study on associations between gut microbial compositions and body mass index \citep{wu2011}, the gut microbial compositions are not measurable and are approximated the microbiota compositions in stool samples. In a typical sequencing study (e.g., 16S rRNA microbiome surveys or single-cell RNA-seq studies) that motivates our research, the observed compositions $O_{ij}$'s are often calculated from observed sequence counts $W_{ij}$'s. That is, $O_{ij}= W_{ij}/N_i$, where $N_i=\sum_{j=1}^p W_{ij}$ is the total counts (or sequencing depth) of sample $i$. These $W_{ij}'s$ are often called original scale of measurements for compositions. By treating these original scale of measurements as random realizations of a certain distribution with compositions being its parameters, it has been argued that modelling these original scale of measurements has advantages over Aitchison's log-ratio approaches \citep{firth2023}. However, this approach \citep{firth2023} often treat compositional measurements as response variables, which does not apply to our setting of regression with corrupted compositional predictors as explanatory variables in log contrast regression models considered in the current article. Since true compositions $U_{ij}$'s may not be directly measurable, a naive idea is to run a surrogate linear log-contrast regression model with observed measurements $W_{ij}$'s:
\begin{equation} \label{sor1}
y_i=\sum_{j=1}^p log(W_{ij}) \beta_j + \epsilon_i, \:  s.t., \: \sum_{j=1}^p \beta_j=0.
\end{equation}
Clearly, existence of measurement errors will cause departure of estimated coefficients $\hat{\bbeta}=(\hat{\beta}_1,\ldots,\hat{\beta}_p)$ of model \eqref{sor1} from the true values $\bbeta^*=(\beta^*_1,\ldots,\beta^*_p)$ of interest in model \eqref{true}. Using a bias correction approach, \citet{shi2022} has shown that the bias term $\hat{\bbeta}-\bbeta^*$ can be improved if $W_{ij}$ in model \eqref{sor1} is replaced by the corrected variables $W^c_{ij}=W_{ij} + \frac{N_i+\alpha_i+1}{2\alpha_i+1}$, where $\alpha_i$ is the over-dispersion parameter associated with sample $i$ in the Dirichlet-Multinominal distribution \citep{shi2022}.  In other words, the authors recommend to use the following model \eqref{sorc} to get more accurate inference on $\bbeta^*$:
\begin{equation} \label{sorc}
y_i=\sum_{j=1}^p log(W^c_{ij}) \beta_j + \epsilon_i, \:  s.t., \: \sum_{j=1}^p \beta_j=0.
\end{equation}

The estimator solved from \eqref{sorc} requires the knowledge of observed counts and leverages the Direchlet-Multinomial distribution to correct observed counts towards unobservable underlying compositions. However, as pointed out in a recent paper \citep{gihawi2023}, sequence read counts could be inflated by many orders of magnitude due to potential contamination issues of draft reference genomes. Assumptions of the previous method \citep{shi2022} may no longer hold when measurement errors in counts are extremely huge. Moreover, it is sometimes less meaningful to analyze counts from different platforms than compositions \citep{allali2017}, which limits the applicability of model \eqref{sorc} to large cohort studies where samples are typically sequenced at different locations/batches using different platforms. Finally, sign consistency property of the variable correction regularized estimator has not been established yet \citep{shi2022}, which further motivates our investigation to take an alternative approach to fill this research gap.

\subsection{Log contrast models with measurement errors} 
A notable nature of measurement errors in compositional covariates is the ripple effect, that is,  measurement error in one component has an impact on at least one of the other components due to the unit-sum constraint on compositions. Unfortunately, most existing high-dimensional error-in-variable regression methods and theories have focused on the unconstrained data \citep{rosenbaum2010,loh2012,datta2017}, which are not directly applicable to compositional data. In a typical sequence count technologies-based microbiome study, discrepancies between $U_{ij}$ and $O_{ij}$ (or simply sequencing bias) are thought to approximately act multiplicatively on the taxon abundances \citep{mclaren2019}. Under this assumption, the following mathematical model has been proposed by computational biologists to characterize the sequence bias issue of microbiome compositional data collected from next generation sequencing experiments \citep{clausen2022,mclaren2022}:
\begin{equation} \label{bias}
O_{ij}=U_{ij} \cdot \frac{E_{ij}}{\sum_{j=1}^p U_{ij}E_{ij}},
\end{equation}
where $E_{ij}$ denotes the measurement error term for sample $i$ and taxon $j$. That is, measurement errors are both sample-specific and taxon-specific. This is because, on the one hand, taxa are not all detected equally well and measurement error is determined by the interaction between experimental protocols and the biological/chemical/physical state of that taxon. On the other hand, samples might come from different sources and be treated by different specimens or technicians and thus measurement error may also depend on the sampling process \citep{mclaren2022}. For these reasons, we assume that measurement errors depend on taxon characteristics and the sampling process, but not on its relative abundance. That is, we assume that $E_{ij}$ and $U_{ij}$ are independent. 

The compositional nature of measurements has been clearly addressed in model \eqref{bias} since $\sum_{j=1}^p O_{ij}=1$. That is, ripple effects of measurement errors in one compositional component to another are well accommodated in model \eqref{bias}. Also, based on model \eqref{bias}, we have $\sum_{j=1}^p \beta_j log(O_{ij})=\sum_{j=1}^p \beta_j log(U_{ij})+\sum_{j=1}^p \beta_j log(E_{ij})$ holds for any $\beta_1,\ldots,\beta_p$ with $\sum_{j=1}^p \beta_j=0$. For ease of presentation, we define $X_{ij}=log(U_{ij}), Z_{ij}=log(O_{ij}),B_{ij}=log(E_{ij})$. Then, in the scale of log contrasts, we have $\bZ \bbeta =(\bX+\bB)\bbeta$, where $\bbeta$ is an arbitrary zero-sum vector. We emphasize that the formulation $\bZ=\bX+\bB$ is only true under the log contrast regression model and this formulation will largely facilitate measure error modelling for regression analysis with contaminated compositional predictors. In other word, we do have ``blessing of compositionality'' in measurement error modelling for log contrast regression with compositional covariates. We further assume that rows of $\bB$ are independent and identically distributed with zero mean, finite covariance $\bSigma_B$ and sub-Gaussian parameter $\tau^2$ for the purpose of developing further statistical inference.

Let $S =\{j : \beta^*_j \neq 0\}$ denote the support of $\bbeta^*$ and $s =|S|$ is the cardinality of $S$. For any subset $J \subset \{1,\ldots, p\}$ and any index $j \in J$, let $J^c$ denote the complement of $J$ and define $J-j = J \backslash\{j\}$. For any matrix $K\in \mathbb{R}^{a\times b}$, $K_{J}$ denotes the submatrix of the $j$th column for $j\in J\cap\{1,\dots, b\} $ of $K$, $K_{J_1J_2}$ is the submatrix formed by $(i,j)$th entries for $i\in J_1\cap\{1,\dots, a\} ,j\in J_2\cap \{1,\dots, b\}$. $\|\cdot\|_{\infty}$ and $\|\cdot\|_1$ denote the $l_\infty$ norm and  $l_1$ norm respectively. Then, our log contrast model is expressed as: $\by =\bX \bbeta^* + \bepsilon = \bX_S\bbeta^*_S + \bepsilon,  \: s.t., \: \sum_{j=1}^p \beta^*_j=\sum_{j \in S} \beta^*_j=0$,  where $\bepsilon=(\epsilon_1,\ldots,\epsilon_n)$ are independent and identically distributed random errors that follow $N(0,\sigma^2)$. While predictors $\bX$ may not be observable, we propose a new method that use its surrogate $\bZ$ to perform inference on $\bbeta^*$. We further introduce some necessary notation before introducing our new method. Let $\bX^p=(log(U_{ij}/U_{ip})) \in R^{n \times (p-1)}$ and $\bZ^p=(log(O_{ij}/O_{ip})) \in R^{n \times (p-1)}$. Without loss of generality, we pick the last component as the reference level when defining $\bX^p$ and $\bZ^p$ in order to facilitate development and presentation of our method. However, in the subsequent methodology development, we will guarantee that the proposed method is invariant on the selection of the reference component. Corresponding to these log ratios, we can partition the $p$-dimensional regression coefficients as $\bbeta=(\bbeta_{-p},\beta_p)$, where $\bbeta_{-p}=(\beta_1,\ldots,\beta_{p-1})$ is free of any constraint. Assuming all columns in $\bX$ are centered and define $\bSigma=\frac{1}{n}\bX^T\bX$. Then, one way to fit a sparse log contrast model in the high-dimensional setting is via the following $L_1$-regularization:
\begin{align}  \label{logconstrast}
\hat{\bbeta}&=  \mathop{\arg\min}_{\bbeta} \frac{1}{2n} \|\by-\bX\bbeta\|_2^2+\lambda\|\bbeta\|_1 \:,  s.t., \: \sum_{j=1}^p \beta_j=0 \nonumber \\
            &=  \mathop{\arg\min}_{\bbeta_{-p}} \frac{1}{2n} \|\by-\bX^p\bbeta_{-p}\|_2^2+\lambda\|\bD^p\bbeta_{-p}\|_1 \nonumber \\ 
            &=  \mathop{\arg\min}_{\bbeta_{-p}} \frac{1}{2} \bbeta_{-p}^T\bSigma^p\bbeta_{-p}-(\brho^p)^T\bbeta_{-p}+\lambda\|\bD^p\bbeta_{-p}\|_1,
\end{align}
where $\bSigma^p=\frac{1}{n}(\bX^p)^T\bX^p$, $\brho^p=\frac{1}{n}(\bX^p)^T\by$ and $\bD^P=(\bI_{p-1},-\mathbf{1}_{p-1})^T$.

Since $\bX$ is unobserved, we have to choose surrogates $\tilde{\bSigma}^p$ and $\tilde{\brho}^p$ to replace $\bSigma^p$ and $\brho^p$ in \eqref{logconstrast}. We begin by first constructing $\hat{\bSigma}$, an unbiased estimator of $\bSigma$. As we specify the observed design matrix $\bZ$ being contaminated by additive measurement error $\bB$, where the rows of $\bB$ are independent and identically distributed with zero mean and finite covariance $\bSigma_B$, then $\hat{\bSigma}=\frac{1}{n}\bZ^T\bZ-\bSigma_B$ is an unbiased estimator of $\bSigma$. Then, following suggestion of CoCoLasso \citep{datta2017}, The surrogate $\tilde{\bSigma}$ and $\tilde{\brho}$ for the unobserved $\bSigma=\frac{1}{n}\bX^T\bX$ and $\brho=\frac {1}{n} \bX^T\by$ are then defined as:
\[
\tilde{\bSigma} = \mathop{\arg\min}_{\bK \ge 0} \|\bK-\hat \bSigma\|_{\max} \:\:, \:\: \tilde{\brho}=\frac{1}{n}\bZ^T\by,
\]
where
$\|\bK\|_{\max}$ denotes the element-wise maximum norm of $\bK$ and $\tilde{\bSigma}$ is obtained through ADMM algorithm. Through matrix transformation, we have $\bSigma^p=\frac{1}{n}(\bX^p)^T\bX^p=(\bD^p)^T\bSigma \bD^p$. Then, surrogates $\tilde{\bSigma}^p$ and $\tilde{\brho}^p$ is calculated from
\[
\tilde{\bSigma}^p=(\bD^p)^T \tilde{\bSigma} \bD^p  \:\:, \:\:  \tilde{\brho}^p=(\bD^p)^T \tilde{\brho}.
\]
Finally, plugging these quantities calculated from observed measurements $\bZ$ into \eqref{logconstrast}, we obtain the following estimator
\begin{equation}\label{eicolasso}
\hat{\bbeta} = \mathop{\arg\min}_{\bbeta \, s.t. \bbeta=\bD^p\bbeta_{-p}}  \frac{1}{2}\bbeta_{-p}^T \tilde{\bSigma}^p \bbeta_{-p} -(\tilde{\brho}^p)^T\bbeta_{-p}+\lambda\|\bD^p\bbeta_{-p}\|_1.
\end{equation}
We term this new estimator \eqref{eicolasso} as the \textbf{Er}ror-\textbf{i}n-\textbf{c}omposition Lasso estimator or simply Eric Lasso hereafter in this paper.

\textbf{\emph{Remark 1}}: 
The core of CoCoLasso lies in the two-step procedure of constructing a good estimator for matrix $\bSigma$: first getting an unbiased estimation $\hat{\bSigma}$ and then conducting projection to obtain a positive semi-definite matrix $\tilde{\Sigma}$. In our context of regression analysis with compositional predictors, the most natural approach is to directly apply this two-step procedure to log-ratio transformed measurements $\bZ^p$ to obtain surrogates $\tilde{\bSigma}^p$ and $\tilde{\brho}^p$ of $\bX^p$, which unfortunately leads to analysis results that are sensitive to selection of the reference component and thus lacks valid interpretation for compositional data analysis. In contrast, the proposed Eric Lasso method is permutation invariant, which gives the same result under any permutation of the $p$ components. Throughout all assumptions and proofs in the subsequent parts of this article, we have intentionally ensured that conditions and conclusions are all independent from the selection of the reference component to enhance statistical interpretations of our compositional data analysis.

\textbf{\emph{Remark 2}}:
The error covariance $\bSigma_B$ might be unknown in practice, and must be obtained through estimation in such a case. Suppose we can borrow information from either independent external data or replicated data to calculate an error matrix $B_o \in \mathbb{R}^{n\times p}$ and correspondingly $\hat\bSigma_B=\frac{1}{n}B_o^TB_o$ as an estimate of $\bSigma_B$, which has also been widely assumed in literature \citep{loh2012,shi2022}. We can show that our theoretical analysis still holds when replacing $\bSigma_B$ with $\hat\bSigma_B$ in Lemma 1 in Section A of the online supplementary materials.

\section{Theoretical analysis} \label{thm}
We assume without loss of generality that $p\in S$ such that $S-p$ is well-defined. Otherwise, we could permute the compositional covariates to ensure it. This assumption is reasonable as long as we can guarantee all conditions and proofs are invariant to permutations of covariates indices, which has been well checked throughout this article. To establish the theoretical results of Eric Lasso, we need the following two regularity conditions:

\begin{con}\label{con:minis}
The matrix $\Sigma_{SS}$ satisfies
\begin{equation}
\Lambda_{\min}(\Sigma_{SS})\ge C_{\min}>0, \label{con:eigen}
\end{equation}
where $\Lambda_{\min}(K)$ denotes the minimal eigenvalue of matrix $K$.
\end{con}

\begin{con}\label{con:mu-inco}
There exists some $\xi \in (0,1]$ such that
    \begin{equation}\label{Term:mu-inco}
        \|\Sigma^p_{S^cS}(\Sigma^p_{SS})^{-1}
        \{\sgn(\beta^*_{S-p})-\sgn (\beta_p^*)1_{s-1}\}+\sgn(\beta_p^*)1_{p-s}\|_{\infty}\le 1-\xi.
    \end{equation}
\end{con}

Condition \ref{con:minis} is a common assumption in high-dimensional statistics and some implications to guarantee the permutation invariance of our analysis under this condition are derived in Lemma 3 and Lemma 4 of the online supplementary materials. Condition \ref{con:mu-inco} is taken from \citet{lin2014}, which is central to guaranteed support recovery of $L_1$ regularization. An important quantity associated with Condition \ref{con:mu-inco} is $\phi = \|D^p_{SS}(\Sigp_{SS})^{-1}(D^p_{SS})^T\|_{\infty}$. According to Proposition 2 of \citet{lin2014}, Condition \ref{con:mu-inco} and $\phi$ are permutation invariant. That is, $\phi$ and the left-hand part of Equation \eqref{Term:mu-inco} are both independent from the selection of reference component $p$. Moreover, \eqref{con:eigen} is permutation invariant naturally. So far, all our conditions are permutation invariant. We further assume that $\max_{j}\|X_j\|_2^2\le n$. With all these regulatory conditions, we now are ready to establish the main theoretical results on our Eric Lasso estimator.

\begin{theorem}[error bound and sign consistency]\label{Th: sign}
Under Condition \ref{con:minis} and Condition \ref{con:mu-inco}, let $\zeta=\max(\tau^4,\sigma^4,1)$, for $\lambda \leq \tau^2$ and $\eps \leq \min(\eps_1,\lambda /(\lambda \eps_2 +\eps_3))$ where $\eps_i$'s are bounded positive constants depending of $\Sigma$, $\beta^*_S$, $\theta$ and $\phi$, Then there exists universal constants $C$ and $c$, with probability at least $1-\delta_1$ where  $\delta_1 = p^2C\exp\left(-cn(s-1)^{-4}\eps^2\zeta^{-1} \right)+pC\exp\left(-cns^{-2}\lambda^2\xi^2\zeta^{-1} \right)$, problem(\ref{eicolasso}) has an optimal solution $\hat\bbeta$ that satisfies the following properties:
(a) $l_{\infty}$-loss: $\linf{\hat\beta_S-\beta^*_S}\le 9\phi\lambda/2$, 
(b) sign consistency: if $\min_{j\in S}|\beta_j^*|>9\phi\lambda/2$, then $\sgn(\hat\beta)=\sgn(\beta^*)$.
\end{theorem}

\textbf{\emph{Remark 3}}:
In CoCoLasso, the parameter $\delta_1$ goes to zero when $s^2\log p/n \to 0$ as $n,p\to\infty$ and we need $(s-1)^4\log p/n \to 0$ in Eric Lasso. This discrepancy is due to the fact that the estimation of $\beta_S$ is actually done via estimation of $\beta_{S-p}$ in a $(s-1)$-dimensional space and then $\beta_p = -\sum_{j\neq p} \beta_j$ is obtained as a transformation. As a result, we have the inequality: $\linf{\hat\beta_{S}-\beta^*_S} \le (s-1)\linf{\hat\beta_{S-p}-\beta^*_{S-p}}$. In the same framework as CoCoLasso, we can only bound $\beta_{S-p}$, leading to an error term with an extra factor of $(s-1)$.

\textbf{\emph{Remark 4}}:
To understand the asymptotic implications of Theorem 1, we assume for simplicity that $\phi$ is constant.
From the expression of $\delta_1$, If $(s-1)^4\log p / n \to 0 \text{  as } n,p \to \infty$ and $\min_{j\in S}|\beta_j^*|\gg s(\zeta \log p / n)^{1/2}$, then we can choose a $\lambda$ satisfying both $\lambda \gg  s(\zeta \log p / n)^{1/2}$ and $\min_{j\in S}|\beta_j^*|>9\phi\lambda/2$ such that $\delta_1$ goes to zero and the sign-consistency of Eric Lasso is achieved.

\section{Simulation Studies} \label{simu}
We have conducted comprehensive numerical studies to evaluate the performance of the proposed Eric Lasso method. The goals of our simulation studies are: 1) to provide numerical evidence supporting the theoretical results established for Eric Lasso estimator in Theorem 1. 2) to assess robustness of the proposed method using simulated data generated from a variety of different scenarios including misspecified models.

\subsection{Simulation setup}
We used three different approaches to generate compositional predictors to evaluate the performance of the proposed method. In the first approach, we followed the simulation design of a previous paper \citep{lin2014} to use the logistic normal distribution \citep{aitchison1980} to simulate unobserved community compositions. In particular, we first simulated  a $n \times p $ latent matrix denoted as $W$ from a multivariate normal distribution $N(\btheta, \bSigma_W)$ where $(\bSigma_W)_{ij}=\rho^{|i-j|}$, where $\rho=0.5$ and $\btheta$ was set in the following way: $\theta_j = \log(0.2 p) $ for j = 1 to 5, while for other indices j, $\theta_j$ was set to 0. The unobserved community compositions were further calculated as $U_{ij} = \exp(W_{ij})/ \sum_{k=1}^p\exp(W_{ik})$ and $X_{ij}=log(U_{ij})$. Compositions generated under this scheme were very heterogeneous in that the first five components dominate the remaining components. The first a few coordinates of the true coefficient vector $\bbeta^*$ were specified as [1.2, -0.8, 0.7, 0, 0, -1.5, -1, 1.4] with remaining components being zeros. The response was then generated as $\by=\bX \bbeta + \bepsilon$, where $\epsilon_1,\ldots,\epsilon_n$ are iid errors simulated from $N(0, 0.5^2)$. Since community compositions are often unobserved, we simulated a corrupted version by adding errors. In particular, we first generated an error matrix $\bB$, whose rows were independently simulated from $N(\bzero, \tau^2\bI)$ and then generated the corrupted compositions $O_{ij}$ according to model \eqref{bias}, which were further transformed into $Z_{ij}=log(O_{ij})$. In the second simulation scenario, we considered more homogeneous community compositions by using the Dirichlet distribution to generate community compositions \citep{fiksel2022}. In particular, each row of $\bU$ was independently simulated from the Dirichlet ($\frac 1 p\bone_p$), where $\bone_p$ denotes the $p$-dimensional vector of all ones and we kept all other settings the same as the first scenario. In the third simulation scenario, we examined method performance under misspecified models violating our core model assumptions as stated in Equation \eqref{bias}. In particular, we followed the exact simulation design of variable correction regularized estimator \citep{shi2022} to generate latent counts. We first generated underlying compositions ($U_{i1},\ldots,U_{ip}$) from the above logistic normal distribution and total sequence counts $N_i$ of sample $i$ from the negative binomial distribution with mean $3\times 10^4$ and variance $3\times 10^6$. Then, observed counts $(W_{i1},\dots,W_{ip})$ were simulated from DirMult$(N_i,\alpha U_{i1},\dots,\alpha U_{ip})$ with $\alpha=5000$, where DirMult denotes Dirichlet-Multinomial distribution. Following suggestion of \citet{shi2022} on handling zero counts, we calculated the observed design matrix $\bZ$ as $Z_{ij}=log((W_{ij}+0.5)/\sum_j(W_{ij}+0.5))$. For ease of presentation, we term these three different data generation mechanisms as Scenario 1, 2 and 3, respectively, hereafter.

Simulation configuration parameters $(n,p,\tau)$ have not been specified yet. We followed simulation designs of two major competitors \citep{lin2014,datta2017} of Eric Lasso to evaluate the performance of our method in a comprehensive manner. In Simulation I, we fixed $\tau=0.5$ and varied $(n,p)=(100,200), (250,400), (500,500)$ to evaluate the performance of our method under different data sizes. This type of setting was considered in the compositional Lasso paper \citep{lin2014}. In Simulation II, we fixed $(n,p)=(100,100)$ to mimic a real data set analyzed later in this paper and varied $\tau$ to evaluate the robustness of the proposed method with respect to the noise level of measurement errors. This type of setting was considered in the CoCoLasso paper \citep{datta2017}. Since Scenario 3 does not contain $\tau$ as a configuration parameter, it was only considered under Simulation I.

After each dataset was simulated under a particular scenario and setting. We applied four different methods to data $(\by,\bZ)$ to obtain the regression coefficient estimator $\hat{\bbeta}$ and then compared it to true regression coefficients $\bbeta^*= (1.2, -0.8, 0.7, 0, 0, -1.5, -1, 1.4,0,\ldots,0)$. The four methods are the proposed Eric Lasso estimator, the compositional data analysis version of Lasso estimator \citep{lin2014}, the CoCoLasso estimator \citep{datta2017} and the vanilla Lasso \citep{tibshirani1996}. Among these methods, compositional Lasso incorporates the compositional nature in data but ignores measurement errors. CoCoLasso accommodates measurement errors yet fails to model the compositional constraint in the regression covariates space. The original Lasso method fails to accommodate neither characteristics of the data and the proposed Eric Lasso method is the only one to take both aspects into account. For ease of presentation, we refer to these Lasso methods as Eric, Coda, CoCo and Vani, respectively, hereafter in this article. The tuning parameter selection in Eric and CoCo Lasso was done by using the  calibrated cross validation method \citep{datta2020} and that in Coda and Vani Lasso was done by using cross validation. 

To compare estimation accuracy of different methods, we calculated the following three metrics based on $\bbeta^*-\hat{\bbeta}$ for each method, including the squared error (SE), the prediction error (PE) and the $l_{\infty}$ loss:
\[SE(\hat\bbeta)=\|\bbeta^*-\hat\bbeta\|_2^2, \:\: PE(\hat\bbeta)=(\bbeta^*-\hat\bbeta)^T \bSigma (\bbeta^*-\hat\bbeta), \:\: l_{\infty}(\hat\bbeta)=\max |\bbeta^*-\hat\bbeta|.\]
To compare the sign consistency on selection of different methods, we calculated the following two metrics based on $sgn(\bbeta^*)-sgn(\hat{\bbeta})$, including the false negative rate (FNR) and false positive rate (FPR). These quantities are commonly used  to assess the sensitivity and specificity of a method and are defined as:
\[
FPR=\frac{\#\{j:\hat{\beta}_j \ne 0 \cap \beta^*_j=0\}}{\#\{j:\beta^*_j=0\}}, \:\:
FNR=\frac{\#\{j:\hat{\beta}_j = 0 \cap \beta^*_j \ne 0\}}{\#\{j:\beta^*_j \ne 0\}}.
\]
For each specific simulation scenario, we repeated 100 times to obtain multiple values of these metrics and reported their mean values along with the standard errors of the mean in the following section.

\subsection{Simulation results}
As discussed in Section \ref{method}, model interpretation is important in compositional data analysis. In particular, an important concern in log contrast regression is the zero-sum constraint on regression coefficients, which can guarantee some basic principles in compositional data analysis and has been well recognized in literature \citep{lin2014,shi2022}. Towards this end, we first examine whether the sum of regression coefficients estimated by different Lasso methods is close to zero. For each simulation run, we calculated the value of $\sum_{j=1}^p \hat{\beta}_j$ and then reported the average value over 100 replicates. A t-test was also performed to examine whether it is significantly different from zero. Results under Scenario 1 of Simulation I are reported in Table 1. As shown in Table 1, both Eric Lasso and Coda Lasso obtain estimated coefficients with a sum very close to zero, and thus provide compositional data analysis with meaningful interpretations. On the other hand, CoCo Lasso and Vanilla Lasso fail to preserve the zero-sum constraint on regression coefficients, which makes it more difficult to interpret the corresponding results in the framework of log contrast regression models for compositional data analysis. We observe a similar phenomenon under other scenarios (reported in Section B of the online supplementary materials). 

\begin{table}[h]
\centering
\caption{The average sum of regression coefficients estimated by different Lasso methods and corresponding t-test p values under Scenario 1 of Simulation I.} \label{tab:1}
\begin{tabular}{llll}
\toprule
(n,p)     & Model & $\sum \beta_j$  & p value   \\ 
\midrule
(100,200) & Eric  & -2.4e-08        & 0.600     \\
          & Coda  & -3.2e-08        & 0.406     \\
          & CoCo  & -3.3e-01        & 1.3e-5    \\
          & Vani  & -4.1e-01        & 1.9e-10   \\
\hline
(250,400) & Eric  & -2.2e-08        & 0.483     \\
          & Coda  & -3.0e-08        & 0.218     \\
          & CoCo  & -4.2e-02        & 0.381     \\
          & Vani  & -1.3e-01        & 0.002     \\
\hline
(500,500) & Eric  & -1.6e-08        & 0.628     \\
          & Coda  & -1.1e-08        & 0.664     \\
          & CoCo  &  6.9e-02        & 0.026     \\
          & Vani  & -1.0e-01        & 4.9e-05  \\ 
\bottomrule
\end{tabular}
\end{table}

We next compare the estimation and selection performance of difference Lasso methods. Results under Scenario 1 and 2 of Simulation I are reported in Table \ref{tab:2} and those under Scenario 3 of Simulation I (model misspecification) are  reported in Table \ref{tab:3}.  As shown in these tables, Eric Lasso and CoCoLasso consistently have better estimation performance (in terms of SE, PE and $l_{\infty}$ loss) than those of Coda and Vanilla Lasso. Under all three scenarios, we observe that estimation errors of Eric Lasso tend to decrease as the sample size $n$ increases, with the only exception being the PE of Eric Lasso under Scenario 3. As for the  $l_{\infty}$-loss of Eric Lasso implicated in Theorem 1, we do observe that it tends to vanish as sample size get larger and larger, even when the data are generated from a misspecified model under Scenario 3. For selection accuracy, the FPR of Coda and Vani Lasso are significantly higher than those of Eric and CoCo Lasso under Scenario 1. A similar phenomenon has been observed in the literature that Coda Lasso tends to select more unnecessary false positives to recover the true model \citep{susin2020,srinivasan2021}. On the other hand, differences in FNR of four methods under Scenario 1 are smaller. Patterns under Scenario 2 are similar to those under Scenario 1, and overall, Eric Lasso has the best selection performance under these two scenarios. When models are misspecified under Scenario 3, a remarkable change is that Eric and CoCo Lasso have much worse FNR compared to Coda and Vani Lasso, which is not surprising given that the FPR of Coda and Vani are two to five times of those of Eric and CoCo. There is no methods that are uniformly better than others in terms of both FPR and FNR under this scenario. In summary, only two methods (Eric and Coda) can lead to valid statistical interpretations for log contrast models. Between this two methods, Eric is consistently better than Coda both in terms of estimation and selection performance across all three scenarios considered in Simulation I.

\begin{table}[h]
\centering
\caption{\small{Comparison of different Lasso estimators under Scenario 1 (top half) and 2 (bottom half) of Simulation I. Mean values (standard errors) of different evaluation metrics are calculated based on 100 simulation replicates.}} \label{tab:2}
\scalebox{0.75}{
\begin{tabular}{@{}llllllll@{}} 
\toprule
(n,p)  & Model & SE       & PE         & $l_{\infty}$  & FPR       & FNR        \\ 
\midrule
    & Eric  & 2.16(0.08)  & 0.77(0.03)  & 0.82(0.02)  & 0.08(0)    & 0.17(0.01)\\   
    & Coda  & 2.91(0.08)  & 1.10(0.04)  & 0.89(0.02)  & 0.11(0.01) & 0.16(0.01)\\
(100,200)& CoCo  & 2.22(0.10)& 0.75(0.03)  & 0.82(0.02)  & 0.09(0)    & 0.16(0.02)\\
    & Vani  & 2.96(0.09)  & 1.03(0.04)  & 0.89(0.02)  & 0.13(0.01) & 0.17(0.02)\\
\hline
    & Eric & 1.06(0.03) & 0.43(0.01) & 0.59(0.01) & 0.05(0) & 0.03(0.01) \\
    & Coda & 1.95(0.03) & 0.83(0.01) & 0.73(0.01) & 0.09(0) & 0.03(0.01)  \\
(250,400)& CoCo & 1.03(0.04) & 0.42(0.01)  & 0.59(0.01) & 0.06(0) & 0.02(0.01) \\
    & Vani & 1.92(0.03) & 0.78(0.02) & 0.72(0.01) & 0.11(0.01) & 0.03(0.01) \\
\hline
    & Eric  & 0.59(0.02) & 0.26(0.01)  & 0.45(0.01)  & 0.03(0)  & 0(0)\\
    & Coda  & 1.54(0.02) & 0.72(0.01)  & 0.65(0.01)  & 0.07(0)  & 0(0)\\
(500,500)&CoCo & 0.57(0.01) & 0.26(0.01)  & 0.46(0.01)  & 0.03(0)  & 0(0)\\
    & Vani  & 1.54(0.02) & 0.70(0.01)  & 0.64(0.01)  & 0.08(0)  & 0(0)  \\
\midrule
\midrule
    & Eric  & 0.07(0)   & 0.73(0.02)  & 0.13(0.01)  & 0.18(0.01)  & 0(0)\\
    & Coda  & 0.09(0.01)& 0.84(0.02)  & 0.16(0.01)  & 0.19(0.01)  & 0(0)\\
(100,200)& CoCo& 0.07(0)   & 0.76(0.02)  & 0.13(0.01)  & 0.20(0.01)  & 0(0)\\  
    & Vani  & 0.09(0.01)& 0.85(0.02)  & 0.16(0.01)  & 0.21(0.01)  & 0(0)\\ 
\hline
    & Eric & 0.06(0.01)  & 0.32(0.01) & 0.13(0.01)  & 0.09(0.01) & 0(0) \\
    & Coda & 0.09(0.01)  & 0.48(0.01) & 0.17(0.01)  & 0.11(0.01) & 0(0) \\
(250,400)& CoCo& 0.06(0.01)& 0.34(0.01) & 0.13(0.01)  & 0.11(0.01) & 0(0) \\
    & Vani & 0.09(0.01)  & 0.49(0.01) & 0.17(0.01)  & 0.13(0.01) & 0(0)\\
\hline
    & Eric  & 0.03(0)  & 0.16(0.01)  & 0.09(0)  & 0.05(0)  & 0(0)\\
    & Coda  & 0.06(0)  & 0.34(0.01)  & 0.13(0)  & 0.07(0)  & 0(0)\\
(500,500)&CoCo & 0.03(0)  & 0.17(0.01)  & 0.09(0)  & 0.06(0)  & 0(0)\\
    & Vani  & 0.06(0)  & 0.34(0.01)  & 0.13(0)  & 0.08(0)  & 0(0)\\
\bottomrule
\end{tabular}
}
\end{table}

\begin{table}[h]
\centering
\caption{Comparison of different Lasso estimators under Scenario 3 of Simulation I. Mean values (standard errors) of different evaluation metrics are calculated based on 100 simulation replicates.} \label{tab:3}
\begin{tabular}{@{}lllllllll@{}} 
\toprule
(n,p)  & Model & SE            & PE            & $l_{\infty}$  & FPR           & FNR        \\ 
\midrule
    & Eric & 2.27(0.05) & 0.86(0.02) & 0.81(0.01) & 0.08(0) & 0.17(0.01) \\
    & Coda & 2.44(0.07) & 0.92(0.03) & 0.83(0.01) & 0.14(0.01) & 0.11(0.01) \\
(100,200) & CoCo & 2.45(0.05) & 0.86(0.02)  & 0.83(0.01)  & 0.08(0)  & 0.22(0.01) \\
    & Vani & 2.43(0.07) & 0.87(0.03) & 0.82(0.01) & 0.16(0.01) & 0.12(0.01) \\
\hline
    & Eric & 2.12(0.04) & 0.99(0.02) & 0.78(0.01) & 0.03(0) & 0.15(0.01)  \\
    & Coda & 2.23(0.04) & 1.03(0.02) & 0.85(0.01) & 0.13(0.01) & 0.02(0.01)  \\
(250,400) & CoCo & 2.22(0.04) & 0.98(0.02) & 0.78(0.01) & 0.03(0) & 0.18(0.01) \\
    & Vani & 2.15(0.04) & 0.97(0.01) & 0.84(0.01) & 0.15(0.01) & 0.01(0) \\
\hline
    & Eric & 1.51(0.04) & 0.73(0.02) & 0.67(0.01) & 0.02(0.01) & 0.04(0.01) \\
    & Coda & 1.90(0.02) & 1.02(0.01) & 0.85(0.01) & 0.12(0)    & 0(0)         \\
(500,500) & CoCo & 1.54(0.04) & 0.72(0.02) & 0.64(0.01) & 0.03(0) & 0.03(0.01) \\
    & Vani & 1.88(0.02) & 1.00(0.01) & 0.85(0.01) & 0.13(0.01) & 0(0)        \\
\bottomrule
\end{tabular}
\end{table}

We next evaluate the robustness of the proposed Eric Lasso method against different noise levels of measurement errors, which is measured by the $\tau$ parameter used in simulation Scenario 1 and 2 of Simulation II. To achieve this goal, we use ROC curves to compare performance of different Lasso methods under different $\tau$ values. The ROC curves of four methods are displayed in Figure \ref{Fig:1}. As the error level increases, it is more difficult to detect findings for all methods such that both FPR and TPR reduce. Clearly, the left panel shows that Eric Lasso and CoCo Lasso have larger areas under curve (AUC) than Coda and Vanilla Lasso method, which is consistent with our conclusions found in Simulation I. Furthermore, the right panel of Figure \ref{Fig:1} shows that Eric has a larger AUC than CoCo Lasso. The same is also true when comparing Coda and Vani. ROC curves under Scenario 2 show the exact same pattern and are displayed in Section B of the online supplementary materials. Therefore, under both scenarios of compositional predictors, the special consideration in Eric Lasso to accommodate compositionality does improve its  discriminative power over the classic CoCo Lasso method developed for high-dimensional Gaussian data. If the comparison is limited to methods with valid log contrast model interpretations, the area under ROC curve of Eric is significantly larger than that of Coda.

\begin{figure}[htb]
\centering
\begin{tabular}{@{}cccc@{}}
   \includegraphics[scale=0.4]{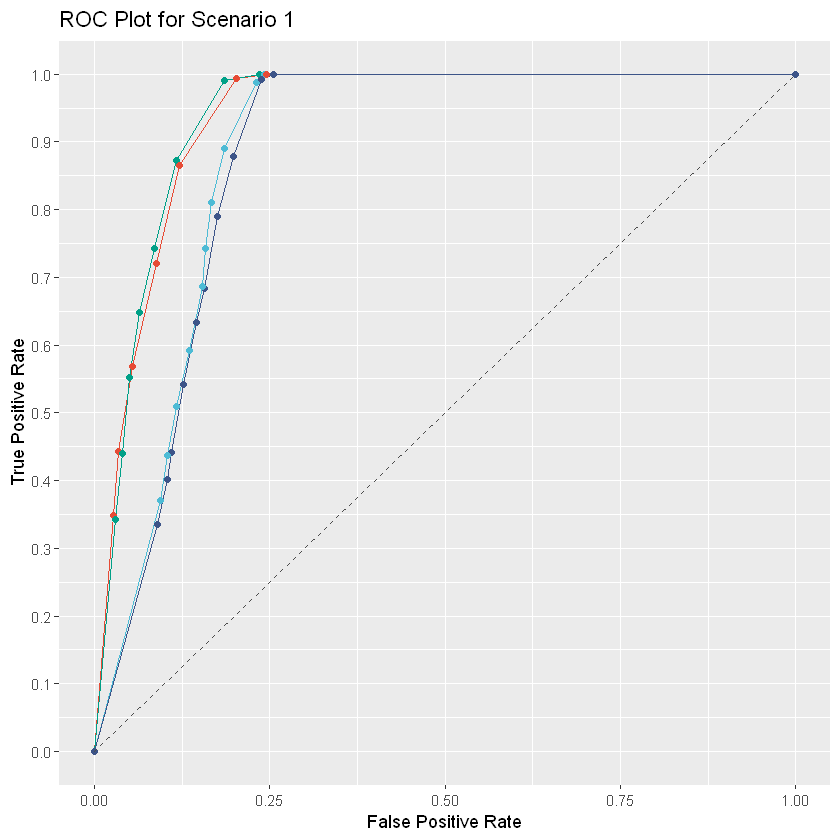}&
   \includegraphics[scale=0.4]{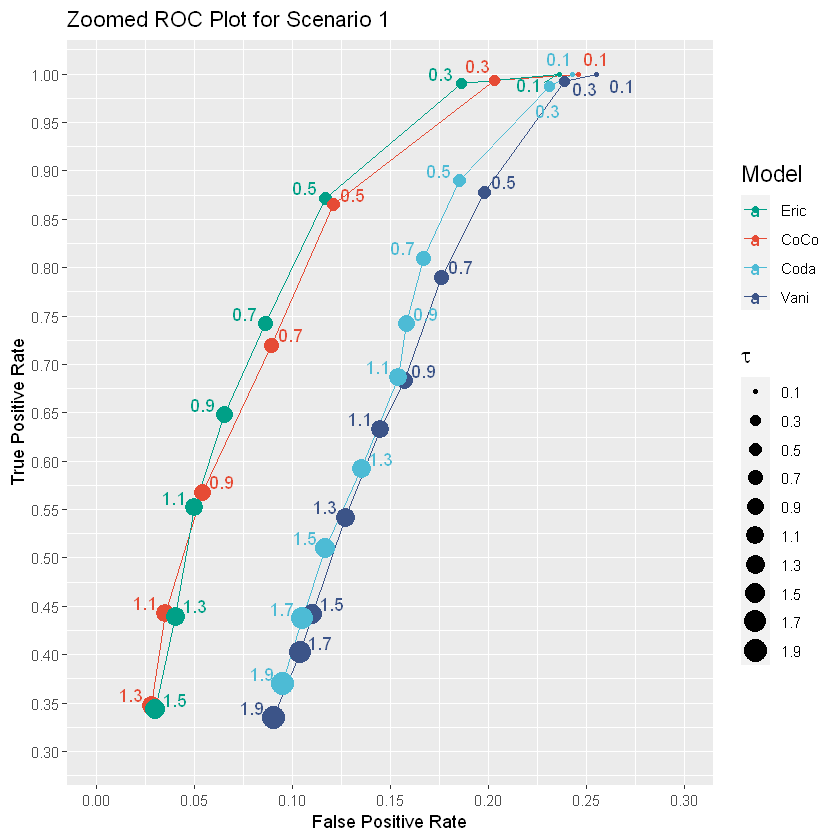}
\end{tabular}
\caption{The ROC curve with different $\tau$ values under Scenario 1. The left panel is ROC curves at original FPR and TPR scales and the right panel zooms in specific regions of FPR and TPR to better distinguish different methods.} \label{Fig:1}
\end{figure}

Combining all numerical results in Simulation I and II, the proposed Eric Lasso method stands out in obtaining a regression model with not only more accurate coefficients estimation and selection, but also valid model interpretation for compositional data regression analysis among the four Lasso methods being evaluated. Between the two methods (Eric and Coda) with valid interpretation for log contrast regression model, the performance of Eric is substantially better than that of Coda. It is of note that all aforementioned results are still true even when data were generated under a misspecified model, which demonstrates the robustness of Eric Lasso. Therefore, our new method is the best candidate for high-dimensional regression analysis of compositional covariates with measurement errors compared to these existing potential competitors.

\section{A case study} \label{rda}
The human gut microbiota has been shown to play a very important role in nutrient digestion and absorption \citep{wu2011} and most existing analyses have successfully shown that there exists a significant association between the overall gut microbiome community and body mass index (BMI) using permutational multivariate analysis of variance type of approaches \citep{tang2017}. However, identifying specific microbial taxa associated with the outcome is more challenging than detecting an existing overall community-level association, partially due to the compositional effect of microbiome data (individual taxa are closely related to or affected by each other). To illustrate the potential usefulness of Eric Lasso, we applied it to investigate BMI associated gut microbial taxa using data collected in the COMBO study \citep{wu2011}. A total of 3068 non-singleton operational taxonomic units (OTUs) were detected in the COMBO study. We first aggregated these OTUs into the genus level and then deleted genera that appeared in less than 2 samples, ending with up $p=80$ genera kept in subsequential statistical analysis. After data filtering and quality control, a total of $n=96$ samples were kept for further analysis and we followed the previous suggestion \citep{shi2022} to transfer these sequence counts into relative abundances for further analysis. To adjust for potential confounding effects, we first regressed BMI on total fat and caloric intake and then took residuals as outcomes for association analysis with gut microbial compositions.

Let $\by=(y_1,\ldots, y_n)$ denote the covariate-adjusted BMI values, $U_{ij}$ denote the relative abundance of the $j$th genus in the $i$th sample, $i=1,\ldots,n, j=1,\ldots,p$. To mimic potential measurement error or bias in the sequencing procedure, we reserved $U_{ij}$ as the true abundances and used perturbation to generate corrupted abundances. In particular, for each sample $i$, we first independently simulated each measurement error factor $e_{ij} \sim Unif(0.1,10)$ and then calculated corrupted abundances as $U_{ij}e_{ij}$, which were further normalized into a compositional vector ($O_{i1},\ldots,O_{ip}$) for downstream analysis. Let $X_{ij}= log(U_{ij}), Z_{ij} = log(O_{ij})$ and $\bX = \{X_{ij}\}, \bZ=\{Z_{ij}\}$ be the corresponding design matrix. Our goal is to use the noisy version $(\by,\bZ)$ to infer the relation of $(\by,\bX)$. To achieve this goal, we reserved $\bX$ from model fitting and used it for evaluation purpose only. Following previous analyses \citep{lin2014,shi2022} on this dataset, we generated bootstrap samples $(\by^b,\bX^b,\bZ^b)$ of size $n/2$ from the full dataset $(\by,\bX,\bZ)$, and then used observations in the bootstrap sample for model training and the other observations not selected in the bootstrap sample for prediction evaluation. This whole procedure (including generating $\bZ$ matrix) was repeated for N=100 times and let $(\by^b,\bX^b,\bZ^b), b=1,\ldots,N$ denote the $b$th  bootstrap sample.

For each bootstrap sample, we fitted the model using $(\by^b, \bZ^b)$ to obtain $\hat\bbeta^b$. We keep track of observations used for prediction evaluation to make sure they are not contained in the bootstrap samples used for model training. Let $C_{-i}$ denote the indices of bootstrap samples that do not contain observation $i$. That is, 
\[ C_{-i}=\{b\in 1,\dots,N| x_i \notin \bX^b\}. \]
Then, the average (over $N$ bootstrap samples) leave-one-out (LOO) squared prediction error on observation $i$ is calculated as $\sum_{b\in C_{-i}}(y_i-X_i\hat\beta^b)^2/|C_{-i}|$, where $|C_{-i}|$ denotes the cardinality of set $C_{-i}$. And the mean squared error over all observations is given by:
\[
MSE_{LOO} = \dfrac{1}{n} \sum_{i=1}^n \dfrac{1}{|C_{-i}|} 
\sum_{b\in C_{-i}}(y_i-X_i\hat\beta^b)^2,
\]
which is used to compare different Lasso estimators. The other metric is the mean absolute error, which can be analogously defined as: 
$$
MAE_{LOO} = \dfrac{1}{n} \sum_{i=1}^n \dfrac{1}{|C_{-i}|} 
\sum_{b\in C_{-i}}|y_i-X_i\hat\beta^b|.
$$

Results on prediction errors are reported in Table \ref{tab:4}, where one can see that the proposed Eric Lasso method has the best prediction performance among all four Lasso methods being considered. Coda Lasso has the second best performance in terms of leave-one-out mean squared prediction error and CoCo Lasso has the second best performance in terms of leave-one-out mean absolute prediction error.

\begin{table}[h]
\centering
\caption{Prediction errors on the gut microbiome data by different Lasso estimators.}
\begin{tabular}{@{}lll@{}}
\toprule
Model & $MSE_{LOO}$ & $MAE_{LOO}$ \\
\midrule
Eric & 32.838 & 4.316 \\
Coda & 34.928 & 4.468 \\
CoCo & 35.580 & 4.379 \\
Vani & 35.988 & 4.450 \\
\bottomrule
\end{tabular}
\label{tab:4}
\end{table}

Recall that results established in Theorem 1 guarantee performance accuracy in both estimation and selection. We thus compare selection results of different Lasso methods. Out of the $N=100$ bootstrap replicates, the frequency of each genus taxon being selected by each Lasso method was calculated and presented in Figure \ref{fig:2}. As can be seen in the figure, Coda and Vani Lasso tend to select more taxa than Eric and CoCo Lasso. For example, there are 1, 6, 1, 4 taxa that are selected by Eric, Coda, CoCo, Vani Lasso, respectively, over 50 times out of the 100 bootstrap replicates. The distribution of taxa relative abundances in this COMBO data is very heterogeneous in that the top 5 most abundant taxa account for 80\% of the total abundances of all $p$ taxa. Hence, the unbalanced compositions generated under Scenario 1 and 3 of Simulation I better mimic the distribution of taxa relative abundances in this COMBO dataset than the more homogeneous case used in Scenario 2 of Simulation I. As observed in Table \ref{tab:2} and Table \ref{tab:3} presented in the previous section, the FPR of Coda and Vani under Scenario 1 and 3 are 2-5 times of those of Eric and CoCo Lasso, which may explain why Coda and Vani Lasso have more findings in this data. In other words, our previous experience in numerical simulations implies that many of these additional findings of Coda and Vani in this COMBO dataset might be false positives. The same phenomenon on spurious findings caused by ignoring measurement errors in microbiome compositional data in statistical analysis has also been observed in the literature \citep{hawinkel2019,gihawi2023}. Finally, the common taxon that is selected more than 50\% of times by all methods is genus \textit{Acidaminococcus} of the \textit{Firmicutes} phylum, which was implicated as important for gut dysbiosis in obese patients \citep{wu2011}. In summary, the proposed Eric Lasso method can overall provide the best model prediction and variable selection performance in this gut microbiome data analysis collected from the COMBO study.

\begin{figure}[h]
\centering
\includegraphics[width=0.8\linewidth]{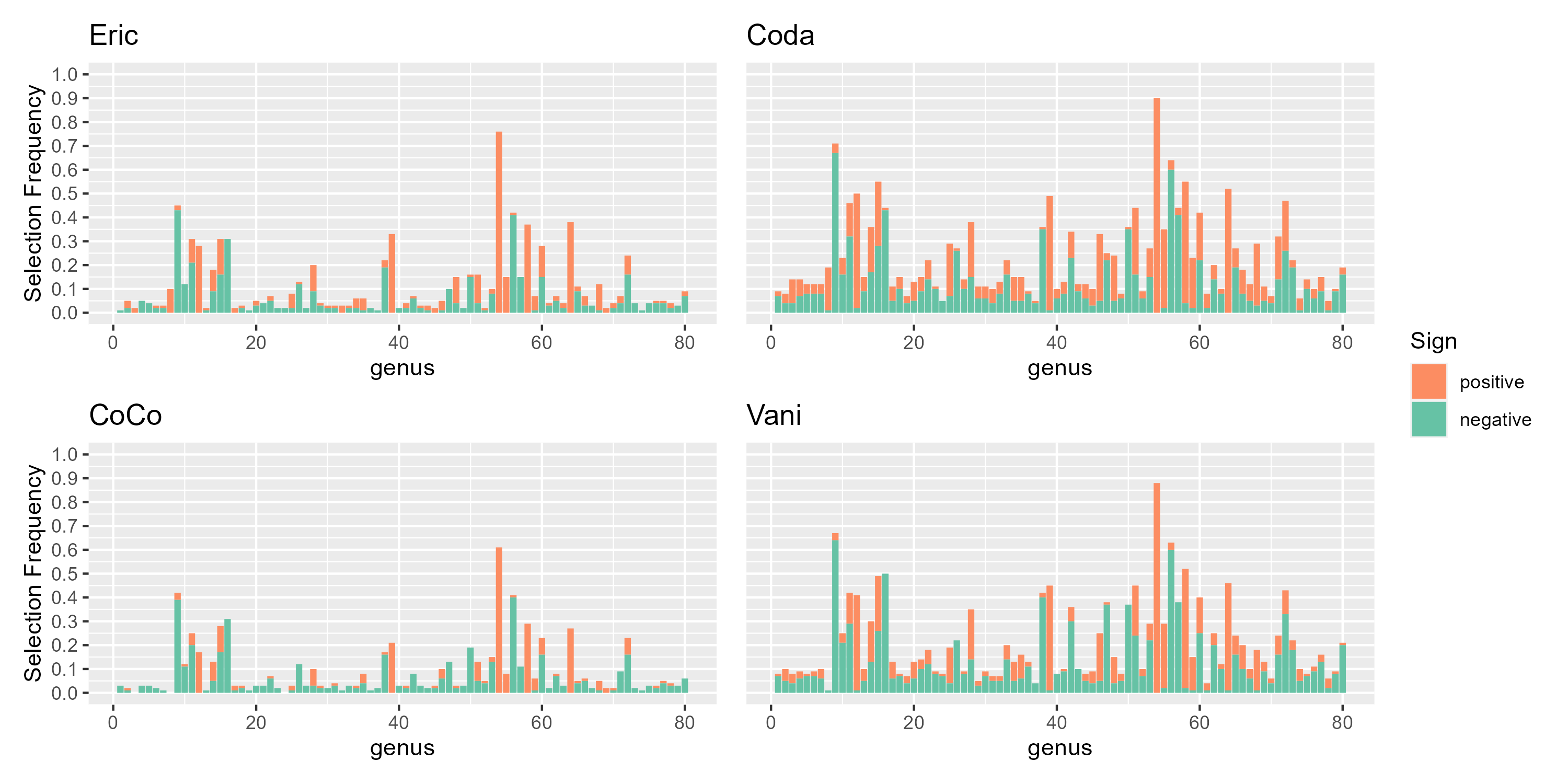}
\caption{Selection frequencies of each taxon under N=100 bootstrap replicates.}
\label{fig:2}
\end{figure}

\section{Discussion} \label{dis}
Motivated by the recent debate on how data contamination could invalidate scientific findings in microbiome research \citep{gihawi2023,poore2020}, we have proposed a novel Eric Lasso approach to combat negative effects of measurement errors in compositional data analysis in order to get more accurate and reliable statistical analysis results. While existence of measurement errors or sequence bias in microbiome compositional data has been widely observed in the literature \citep{mclaren2019,mclaren2022,zhao2021}, statistical methods to address this issue in the high-dimensional regression framework are less developed. As far as we are concerned, the only existing method tailored for this specific problem is the variable correction regularized estimator \citep{shi2022}, which treats the observed counts (contaminated copy) as realizations from an underlying Dirichlet-Multinomial distribution whose parameters are unobserved compositions (clean copy). A major limitation of the variable correction regularized estimator method is requirement on knowledge of total counts of each sample, which however, can be mistakenly measured by many orders of magnitude in practice \citep{gihawi2023}. To gain robustness against such a huge extent in data contamination, we adopt an alternative approach which directly targets at compositions rather than counts. As shown in our numerical studies (i.e., Scenario 3 of Simulation I and the case study in the previous section), our new Eric Lasso approach has a better overall performance than the variable correction regularized estimator method and the improvement in reducing false positive findings is substantial when the underlying true compositions are heterogeneous. While we have only employed the Eric Lasso method to microbiome compositional data analysis to illustrate its potential usefulness and superior performance in this article, the proposed methodology is very general and can be applied to other types of high-dimensional compositional data as well. Finally, similar to previous publications in this research vein \citep{loh2012,datta2017}, we assume the covariance matrix of measurement errors is known or can be estimated from independent sources such as external data or multiple replicates of the data. In the future work, we intend to explore different backup strategies (e.g.,resampling techniques) when such information is less available.

\bibliographystyle{chicago}
\bibliography{myref}

\clearpage
\noindent  {\huge Supplementary Materials for ``High-dimensional log contrast models with measurement errors''}

\setcounter{subsection}{0}
\setcounter{equation}{0}
\setcounter{table}{0}
\setcounter{figure}{0}
\renewcommand\thesection{\Alph{section}}
\renewcommand{\theequation}{S.\arabic{equation}}
\renewcommand{\thefigure}{S.\arabic{figure}}
\renewcommand{\thetable}{S.\arabic{table}}

\section*{A Proofs}
For ease of presentation, we first introduce the following notations:
$$
\begin{array}{cc}
\tilde A=D^p_{SS}(\tsigp_{SS})^{-1}(D^p_{SS})^T & A=D^p_{SS}(\Sigp_{S S})^{-1}(D^p_{SS})^T \\
E=\tsigp-\Sigp & F_{SS}=(\tsigp_{SS})^{-1}-(\Sigp_{SS})^{-1}\\  
G=\Sigp_{S^cS}(\Sigp_{SS})^{-1} & \tilde G = \tsigp_{S^cS}(\tsigp_{SS})^{-1}\\
H=\tilde G-G & b=\|\beta_S^*\|_{\infty}\\
\psi=\|\Sigma_{SS}\|_{\infty} &  \nu=\linf{\Sigma_{S}}\\
\end{array}\\
$$
$$
\Pi = I-\frac 1n \Xp_S(\Sigp_{SS})^{-1}(\Xp_S)^T
$$
We adopt the identical closeness condition as used in CoCoLasso \citep{datta2017}.

\begin{con}[Closeness condition] \label{con:close} 
Let us assume that the distribution of  $\hat \Sigma$ and $\tilde\rho$ are identified by a set of parameters $\theta$. Then there exists universal constants $C$ and $c$, and positive functions $\zeta$ and $\epsilon_0$ depending on $\beta^*_S$, $\theta$ and $\sigma^2$ such that for every $\epsilon \le \epsilon_0$, $\hat\Sigma$ and $\tilde\rho$ satisfy the following probability statements:
\begin{equation}\label{Eq: close}
\begin{array}{c}
\rm pr(|\hSig_{ij}-\Sigma_{ij} | \ge \eps) \leq C \exp \left(-cn\eps^2\zeta^{-1} \right) \;  \; i,j=1,\ldots,p, \\
\rm pr(|\tilde \rho_j-\rho_j | \geq \eps) \leq C \exp \left(-cns^{-2}\epsilon^2\zeta^{-1} \right) \;  \; j=1,\ldots,p.
\end{array}
\end{equation}
\end{con}

As indicated in CoCoLasso, the closeness condition holds for $\hSig$ with $\zeta=\max(\tau^4,\sigma^4,1)$ and $\eps_0=\tau^2$. In scenarios where $\Sigma_B$ is unknown and is substituted by $\hat\Sigma_B=\frac{1}{n}B_o^TB_o$, this condition still satisfies with the same parameter values as guaranteed by the following Lemma \ref{lem: unknown Sig_B}:

\begin{lemma}\label{lem: unknown Sig_B}
$\hSig_o=\frac{1}{n}Z^TZ-\hSig_B$ and $\trho$ satisfy the closeness condition \eqref{Eq: close} with $\zeta=\max(\tau^4,\sigma^4,1)$ and $\eps_0=\tau^2$. 
\end{lemma}
\begin{proof}
Since $\hSig$ satisfies the closeness condition, we have
$$\rm pr(|\hSig_{ij}-\Sigma_{ij} | \ge \eps) \leq C \exp \left(-cn\eps^2\zeta^{-1} \right).$$
Using Lemma B.1 of \citet{datta2017},
$$
\rm pr(|(\hSig_B-\Sigma_B)_{ij} | \ge \eps) \leq C \exp \left(-cn\eps^2\tau^{-4} \right)\leq C \exp \left(-cn\eps^2\zeta^{-1} \right),
$$
$$
\hSig_o-\Sigma=\hSig_o-\hSig+(\hSig-\Sigma)=\Sigma_B-\hSig_B+(\hSig-\Sigma).
$$
Hence by redefining $C$ and $c$, the sub-Gaussianity of $B_o$ implies
$$
\rm pr(|(\hSig_o)_{ij}-\Sigma_{ij} | \ge \eps) \leq C \exp \left(-cn\eps^2\zeta^{-1} \right),
$$
which completes the proof.
\end{proof}

We will now proceed to prove Theorem 1 of the main text under the assumption that the closeness condition is satisfied. To establish major results presented in Theorem 1, we begin by proving the following lemmas.

\begin{lemma}\label{lem:trick} 
For any $\epsilon > 0$ we have
\begin{equation}\label{eq:prick}
\rm pr(\| \tsigp-\Sigp \|_ {\max} \ge \epsilon) \le p^2 \max_{k,l} \rm pr(|\hat \Sigma_{kl}-\Sigma_{kl} | \ge \epsilon/8).
\end{equation}
\end{lemma}

\begin{proof}
Plugging $$D^p=\begin{pmatrix}
I_{p-1} \\
-1_{p-1}
\end{pmatrix}
$$
into $\tilde{\Sigma}^p=(D^p)^T \tilde{\Sigma} D^p$, we have 
$$
(\tsigp)_{kl}=\tilde \Sigma_{kl}-\tilde \Sigma_{kp}-\tilde \Sigma_{pl}+\tilde \Sigma_{pp}.
$$
As $\Sigp=D_p^T\Sigma D_p$, this further implies that
$$
\| \tsigp-\Sigp \|_ {\max}\le 4 \|\tilde \Sigma-\Sigma\|_{\max}.
$$
By definition of $\tilde{\Sigma}$ and the positive semi-definite nature of $\Sigma$, we have
$$\|\tilde\Sigma-\Sigma\|_{\max}\le\|\tilde\Sigma-\hat\Sigma\|_{\max}+ \|\hat\Sigma-\Sigma\|_{\max}\le 2\|\hat\Sigma-\Sigma\|_{\max},$$
\[
\rm pr(\| \tilde \Sigma -\Sigma \|_ {\max} \geq \epsilon/4) \le \rm pr(\| \hat \Sigma -\Sigma \|_ {\max} \ge \epsilon/8). \]
The proof then follows using union bounds over $\rm pr(|\hat\Sigma_{kl}-\Sigma_{kl} | \geq \epsilon/8)$.
\end{proof}

The closeness condition requires that $\hat{\Sigma}$ and $\tilde{\rho}$ are sufficiently close to
$\Sigma$ and $\rho$, respectively. A direct implication of Lemma \ref{lem:trick} is  
\[
\| \tilde{\Sigma}^p- \Sigma^p \|_ {\max}\le 4 \| \tilde \Sigma-\Sigma\|_{\max} \le 8\|\hat\Sigma-\Sigma\|_{\max}, 
\]
which implies that $\tilde{\Sigma}^p$ is sufficiently close to $\Sigma^p$. Analogously, since ($\tilde{\rho}^p -\rho^p) = (D^p)^T(\tilde{\rho}-\rho)$, the closeness condition of $\tilde{\rho}$ on $\rho$ ensures that $\tilde{\rho}^p$ approximates $\rho^p$ well.

\begin{lemma}\label{Lem:eigen}
    $\Lambda_{\min}(\Sigma_{SS}) \le \Lambda_{\min}(\Sigp_{SS})$.
\end{lemma}
\begin{proof} It can be shown that $\|\Dsp v\|^2_2 \ge \|v\|^2_2$ for any $v\in\mathbb{R}^{s-1}\backslash\{0\} $. Hence we have
$$
\begin{aligned}
 \Lambda_{\min}(\Sigma_{SS})&=\frac 1n\min_{v\in\mathbb{R}^s\backslash\{0\}}\dfrac{\|X_Sv\|^2_2}{\|v\|^2_2}\\
 &\le \frac 1n\min_{v\in\mathbb{R}^{s-1}\backslash\{0\}} \dfrac{\|X_S\Dsp v\|^2_2}{\|\Dsp v\|^2_2}\\
 &\le \frac 1n\min_{v\in\mathbb{R}^{s-1}\backslash\{0\}} \dfrac{\|X_S^p v\|^2_2}{\|v\|^2_2}\\
 &=\Lambda_{\min}(\Sigp_{SS}).
 \end{aligned}
$$  
\end{proof}

\begin{lemma}\label{Lem: inv}
$\rm pr(\tsigp_{SS} > 0) \ge 1-Cp^2\exp\left(-cn(s-1)^{-2}\eps^2\zeta^{-1}\right)$ for all $\eps \le min(\eps_0,C_{\min}/16)$.
\end{lemma}
\begin{proof}
Using Lemma \ref{Lem:eigen},
\begin{align*}
 \Lambda_{\min}(\tsigp_{SS}) \ge & \Lambda_{\min}(\Sigp_{SS})- |\Lambda_{\max}(-E_{SS})|  \geq C_{\min}-\|E_{SS}\|_2 \\
 \geq& C_{\min} - (s-1) \|E_{SS} \|_ {\max} \geq C_{\min} - (s-1) \|E\|_ {\max} \geq C_{\min}/2,
\end{align*}
where the last inequality occurs when $\|E\|_{\max} \le 8(s-1)^{-1}\eps$ with probability at least $1-Cp^2\exp\left(-cn(s-1)^{-2}\eps^2\zeta^{-1}\right)$ for $\eps \leq min(\eps_0,C_{\min}/16)$, according to the closeness condition and Lemma \ref{lem:trick}.
\end{proof}

\begin{lemma}\label{lem:gauss}
\begin{equation}
 \|\frac 1n (X_S)^Tw\|_{\infty} \le\lambda/2
\end{equation}
holds with probability at least $1-s\exp\left(-n\lambda^2/(8\sigma^2)\right)$, and
\begin{equation}
\frac 1n\linf{(\Xp_{S^c})^T\Pi w} \le \lambda \xi/4
\end{equation}
holds with probability at least $1-(p-s)\exp\left(-n\lambda^2\xi^2/(128\sigma^2)\right)$.
\end{lemma}
\begin{proof}
This lemma is taken from the proof of \citet{lin2014} and these inequalities hold as a result of the Gaussian tail bound.
\end{proof}

\subsection*{Proof of Theorem 1}
For the optimality, we aim to prove \eqref{goal:S} and \eqref{goal:sub}.
\begin{align}
(\tsigp_{\hat S\hat S},\tsigp_{\hat S\hat S^c})\hat \beta_{-p}-\trhop_{\hat S}+\lambda\{\sgn(\hat \beta_{\hat S-p})-\sgn(\hat \beta_p)1_{s-1}\}&=0,\label{goal:S}\\ 
\|(\tsigp_{\hat S^c\hat S},\tsigp_{\hat S^c\hat S^c})\hat \beta_{-p}-\trhop_{\hat S^c}-\lambda\sgn(\hat \beta_p)1_{p-s}\|_{\infty}&\le \lambda. \label{goal:sub}
\end{align}

\subsection*{Part I: proof of (\ref{goal:S})}
Let $u$ denote $ \{\sgn(\hbetasp)-\sgn(\hat \beta_p)1_{s-1}\}= (\Dsp)^T\sgn(\hat\beta_S)$ for simplicity. Conditioned on Lemma \ref{Lem: inv}, we establish $\hat \beta$ such that
\begin{align}
\hat \beta_{S-p}&=(\tsigp_{SS})^{-1} (\trhop_{S}-\lambda u),\label{Eq:constructbeta}\\ 
\hat\beta_{S^c}&=0,
\end{align}
where \eqref{Eq:constructbeta} is solved from \eqref{goal:S} by replacing $S$ with $\hat{S}$. We rewrite it as
\begin{align*}
\hat\beta_S-\beta^*_S= &\Dsp (\tsigp_{SS})^{-1}\left(\trhop_S-\rhop_S+\Sigp_{SS} \beta_{S-p}^*+\frac 1n (X^p_S)^Tw -\lambda u\right) -\Dsp\beta_{S-p}^* \\            
=& \Dsp F_{SS}(\Dsp)^T\left(\trho_S-\rho_S+\Sigma_{SS} \beta_{S}^* + \frac 1n (X_S)^Tw\right)
+ A (\trho_S-\rho_S) + \frac  1n  A(X_S)^Tw-\lambda \tilde A\sgn(\hat \beta_S).\\
\end{align*}

Let $\eta_1= \|\Dsp F_{SS}(\Dsp)^T\|_{\infty}, \eta_2=\|\Dsp E_{SS}(\Dsp)^T\|_{\infty}$. As $ (\Dsp)^T \Dsp =I_{s-1}$, 
we have
$$
\eta_1=\|\Dsp (\tsigp_{SS})^{-1} (\Dsp)^T \Dsp E_{SS}(\Dsp)^T A\|_{\infty} \le (\eta_1+\phi)\eta_2\phi.
$$
Therefore, by the Lemma \ref{lem:trick} and the closeness condition, for $\eps \le \min (\eps_0, (2\phi)^{-1})$, 
\begin{equation}\label{Eq:eta2}
    \eta_2 \le 2(s-1)^2\|E\|_{\max}\le \eps
\end{equation}
with probability at least $1-Cp^2\exp\left(-\frac{cn\eps^2}{256(s-1)^4\zeta}\right)$. Then
\begin{equation}\label{Eq:eta1}
\eta_1\le \dfrac{\phi^2\eta_2}{1-\eta_2\phi}\le \dfrac{\phi^2\eps}{1-\eps\phi}\le \phi,
\end{equation}
\begin{equation}
\|\tilde A\|_{\infty}\le \phi + \eta_1 \le 2\phi.
\end{equation}

The closeness condition and Lemma \ref{lem:gauss} imply that for $\lambda \le \eps_0$
\begin{equation}\label{Eq:rholam}
\|\trho_S-\rho_S\|_{\infty} \le \lambda, 
\end{equation}
\begin{equation}\label{Eq:Xw}
\|\frac 1n (X_S)^Tw\|_{\infty} \le\lambda/2,
\end{equation}
with probability at least
$1-Cs\exp \left(-cns^{-2}\lambda^2\zeta^{-1} \right)-s\exp\left(-n\lambda^2\sigma^{-2}/8\right)$.
Combining all these inequalities, for $\eps \le \min\{\lambda \phi^{-1}(3\lambda+b\psi)^{-1},\eps_0,\phi^{-1}/2\} $, we have
\begin{align*}
\|\hat\beta_S-\beta^*_S\|_{\infty}\le &\dfrac{\phi^2\eps}{1-\eps\phi}(\lambda +\lambda/2+b\psi)+\phi \|\trho_S-\rho_S\|_{\infty} +\phi \|\frac 1n (X_S)^Tw\|_{\infty} +2\phi\lambda\\
    \le &  9\lambda \phi/2.
\end{align*}
Then we can conclude that $\sgn(\hat\beta_{S})=\sgn(\beta^*_S)$ and \eqref{goal:S} hold.

\subsection*{Part II: proof of (\ref{goal:sub})}
Using \eqref{Eq:constructbeta} and substituting $u$ with $\sgn(\hbetasp)-\sgn(\hat \beta_p)1_{s-1}$, we have 
$$
(\tsigp_{S^cS},\tsigp_{S^cS^c})\hat \beta_{-p}-\trhop_{S^c}-\lambda\sgn(\hat \beta_p)1_{p-s}
=  \tilde G(\trhop_S-\lambda u)-\trhop_{S^c}-\lambda\sgn(\hat \beta_p)1_{p-s}.
$$
Taking the absolute values and using triangular inequalities, we have
$$
\linf{\tilde G(\trhop_S-\lambda u)-\trhop_{S^c}-\lambda\sgn(\hat \beta_p)1_{p-s}}\\
\le  \linf{\tilde G\trhop_S-\lambda Hu-\trhop_{S^c}}+\lambda\linf{Gu+\sgn(\hat \beta_p)1_{p-s}}.
$$
Condition 2 of the main text implies $\linf{Gu+\sgn(\hat \beta_p)1_{p-s}} \le 1-\xi$.
We break the first term as 
$$
\lambda Hu+\trhop_{S^c} -\tilde G\trhop_S = H(\lambda u-\trhop_S)+\{(\trhop_{S^c}-\rhop_{S^c})-G(\trhop_S-\rhop_S)\} + (\rhop_{S^c}-G\rhop_S),
$$
where $\linf{\rhop_{S^c}-G\rhop_S}=\frac 1n\linf{(\Xp_{S^c})^T\Pi w} \le \lambda \xi/4$
with probability at least  $1-(p-s)\exp\left(-n\lambda^2\xi^2/(128\sigma^2)\right)$, according to Lemma \ref{lem:gauss}. We further bound $G(\Dsp)^T$ as
\begin{equation}\label{Eq:Grho}
\begin{aligned}
\linf{G(\Dsp)^T}&=\linf{\frac 1n (\Xp_{S^c})^TX_SA} =\max_{j\in S^c}\|\frac{1}{n}(X_j-X_p)^TX_S A\|_1\\
& \le 2\max_{j\in \{S^c,p\}}\|\frac{1}{n}X_j^TX_SA\|_1 = 2\linf{\Sigma_{\{S^c,p\},S} A}\\
& \le 2\nu\phi. 
\end{aligned}
\end{equation}

We assume that $\eps\le \frac{\lambda\xi}{4+4\nu\phi}$, hence 
$$
\begin{aligned}
\|(\trhop_{S^c}-\rhop_{S^c})-G(\trhop_S-\rhop_S)\|_{\infty} &=\|(\trhop_{S^c}-\rhop_{S^c})-G(\Dsp)^T(\trho_S-\rho_S)\|_{\infty}\\
&\le (2+2\nu \phi)\linf {\trho-\rho} \le \lambda \xi/2
\end{aligned}
$$
holds with probability at least $1-Cp\exp\{-cns^{-2}\eps^2\zeta^{-1}\}$. Since the term $H(\lambda u-\trhop_S)=H(\Dsp)^T\left(\lambda \sgn(\hat\beta_S)-\trho_S\right)$, we have:

$$
\begin{aligned}
    \linf{H(\Dsp)^T}&=\linf{\Sigp_{S^cS}F_{SS}(\Dsp)^T+E_{S^cS}(\tsigp_{SS})^{-1}(\Dsp)^T}\\
    &\le\linf{\Sigp_{S^cS}F_{SS}(\Dsp)^T}+\linf{E_{S^cS}(\tsigp_{SS})^{-1}(\Dsp)^T}\\
    &=\linf{\frac 1n (\Xp_{S^c})^TX_S(A-\tilde A)}+\linf{E_{S^cS}(\Dsp)^T\tilde A}\\
    &\le\linf{\frac 1n (\Xp_{S^c})^TX_S(A-\tilde A)}+2\phi \linf{E_{S^cS}(\Dsp)^T}\\
    &\le 2\nu\eta_1 + 4(s-1)\phi \|E\|_{\max},
\end{aligned}
$$
where the last inequality follows from the analogous trick in \eqref{Eq:Grho}. As mentioned in \eqref{Eq:Xw},we have $\linf{\rho_S}\le \linf{\frac 1n X_S^T w} + \linf{\Sigma_{SS}\beta^*_{S}}\le \lambda/2+b\psi$. Along with \eqref{Eq:rholam}, we have
$$
\linf{\lambda \sgn(\hat\beta_S)-\trho_S} \le \linf{\rho_S}+\linf{\trho_S-\rho_S}+\lambda 
\le 5\lambda/2+b\psi.
$$

Note that \eqref{Eq:eta2} and \eqref{Eq:eta1} yield that $\eta_1\le 2\phi^2 \eta_2\le 4(s-1)^2\phi^2\|E\|_{\max}\le 2\phi^2\eps$. For $\eps \le \frac{\lambda \xi}{(10\lambda+4\psi)(4\nu\phi^2+2\phi)}$,
$$
\begin{aligned}
    \linf{H(\lambda u-\trhop_S)}&\le (5\lambda/2+b\psi)(2\nu\eta_1 + 4(s-1)\phi \|E\|_{\max}) \\
    &\le (5\lambda/2+b\psi)[4\nu\phi^2\eps+2\phi\eps/(s-1)]\\
    &\le (5\lambda/2+b\psi)(4\nu\phi^2+2\phi)\eps\\
    &\le \lambda\xi/4.
\end{aligned}
$$

Gathering them all together, we conclude that for 
$$\eps \le \min \{\frac{1}{2\phi},\eps_0,\frac{\lambda}{\phi(3\lambda+b\psi)},\frac{\lambda \xi}{(10\lambda+4\psi)(4\nu\phi^2+2\phi)}, \frac{\lambda\xi}{4+4\nu\phi},\frac{C_{\min}}{16}\}$$ and $\lambda \le \eps_0$, equations \eqref{goal:S} and \eqref{goal:sub} hold with probability at least $1-C_1p^2\exp\left(-c_1n(s-1)^{-4}\eps^2\zeta^{-1}\right)-C_2p\exp\left(-c_2ns^{-2}\lambda^2\xi^2\zeta^{-1}\right).$

\section*{B Additional simulation results}
In Table 1 of the main text, we have examined the sum of regression coefficients $\sum_{j=1}^p \hat{\beta}_j$ estimated by different Lasso methods under Scenario 1 of Simulation I. The corresponding results under  Scenario 2 and Scenario 3 of Simulation I are reported in Table \ref{tab:s1} and Table \ref{tab:s2}, respectively. As observed in Table 2 of the main text, the sum of regression coefficients estimated by Eric Lasso and Coda Lasso is very close to zero, and the departure from zero for CoCo Lasso or Vani Lasso is more substantial. It is of interest to observe that test results of CoCo Lasso and Vani Lasso may not be significant under certain scenarios, especially under Scenario 2 (i.e., in Table S.1). We take a closer look at the sum of regression coefficients ($\sum_{j=1}^p \hat{\beta}_j$) estimated by CoCo Lasso and Vanilla Lasso and presented histograms of the 100 replicates under each scenario in Figure S.1. It can be seen that the spread of each empirical distribution is relatively large compared to its mean, leading to an insignificant p-value.  Like what have been observed in the main text, the obvious departure of sum of regression coefficients from zero for CoCoLasso and Vani Lasso makes it difficult to interpret the estimated regression model coefficients under the framework of compositional data analysis as explained in Section 2 of the main text.

\begin{table}[h]
\centering
\caption{The average sum of regression coefficients estimated by different Lasso methods  and corresponding t-test p values under Scenario 2.} \label{tab:s1}
\begin{tabular}{llll}
\toprule
(n,p)     & Model & $\sum \beta_j$  & p value \\ 
\midrule
(100,200) & Eric  & -2.2e-08        & 0.924   \\
          & Coda  & 3.5e-08         & 0.847   \\
          & CoCo  & 1.9e-02         & 0.305   \\
          & Vani  & 1.8e-02         & 0.343   \\
\hline
(250,400) & Eric  & -5.0e-07        & 0.269   \\
          & Coda  &  6.8e-08        & 0.846   \\
          & CoCo  & -8.0e-03        & 0.650   \\
          & Vani  & -1.2e-02        & 0.508   \\
\hline
(500,500) & Eric  & 2.9e-07         & 0.532   \\
          & Coda  & 1.9e-07         & 0.530   \\
          & CoCo  & 4.0e-03         & 0.695   \\
          & Vani  & 5.0e-03         & 0.631  \\ 
\bottomrule
\end{tabular}
\end{table}

\begin{table}[h]
\centering
\caption{The average sum of regression coefficients estimated by different Lasso methods  and corresponding t-test p values under Scenario 3.} \label{tab:s2}
\begin{tabular}{@{}llll@{}}
\toprule
(n,p)     & Model & $\sum \beta_j$  & p value   \\ 
\midrule
(100,200) & Eric  & 6.4e-08         & 0.047     \\
          & Coda  & 6.2e-08         & 0.047     \\
          & CoCo  &-6.5e-01         & 4.7e-24   \\
          & Vani  &-1.5e-01         & 0.011     \\
\hline
(250,400) & Eric  & 1.4e-08         & 0.793     \\
          & Coda  &-3.0e-08         & 0.138     \\
          & CoCo  &-4.8e-02         & 9.9e-33   \\
          & Vani  & 5.0e-02         & 0.238     \\
\hline
(500,500) & Eric  &-4.1e-08         & 0.549     \\
          & Coda  &-1.7e-09         & 0.930     \\
          & CoCo  &-3.8e-02         & 5.0e-29 \\
          & Vani  & 8.6e-2          & 0.006     \\ 
\bottomrule
\end{tabular}
\end{table}

\begin{figure} [h]
\centering
\includegraphics[width=0.8\linewidth]{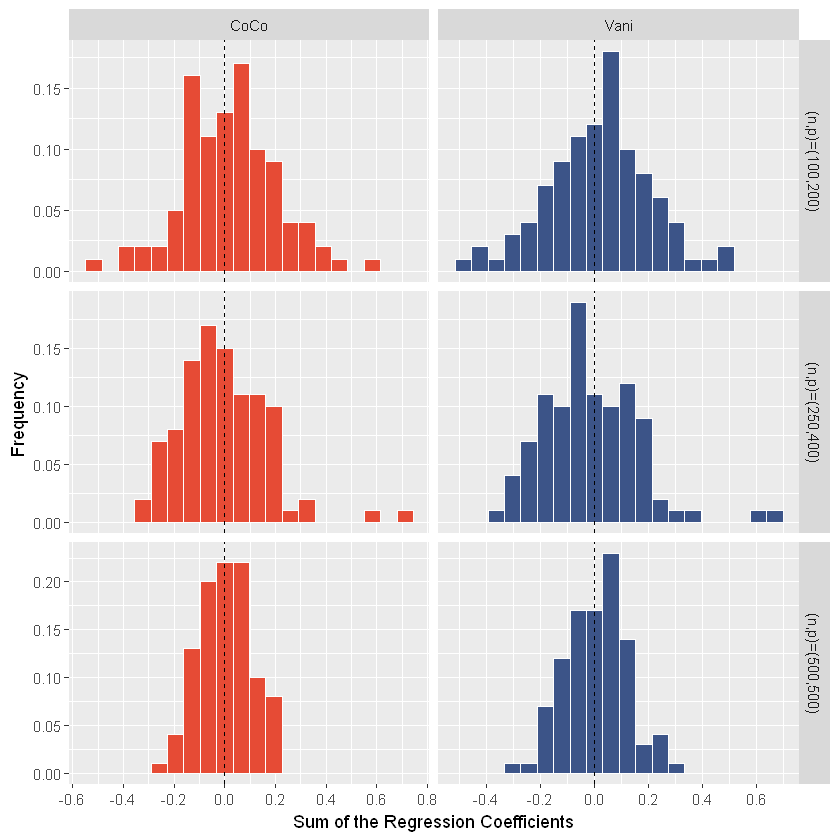}
\caption{The histograms of the sum of coefficients under Scenario 2 for model CoCoLasso and Vanilla Lasso.} \label{fig:hists}
\end{figure}

\clearpage
We next report ROC curves under Scenario 2 of Simulation II in Figure S.2. Similar to patterns displayed in Figure 1 of the main text, Eric Lasso tends to have the best performance among the four methods under Scenario 2 of Simulation II.

\begin{figure}[htb]
\centering
\begin{tabular}{@{}cccc@{}}
\includegraphics[scale=0.36]{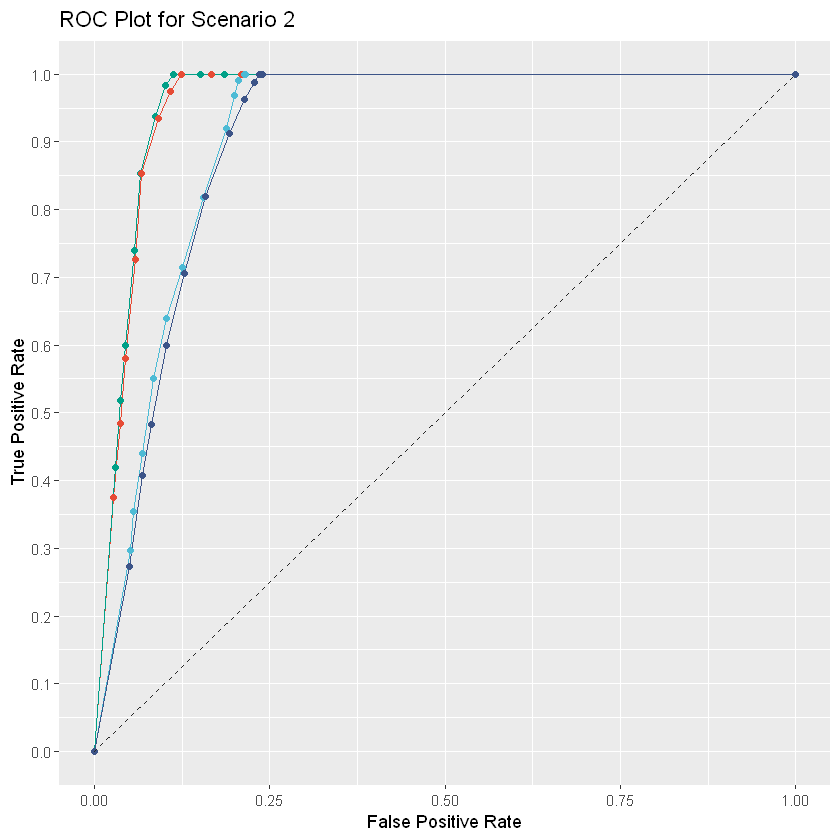} &
\includegraphics[scale=0.36]{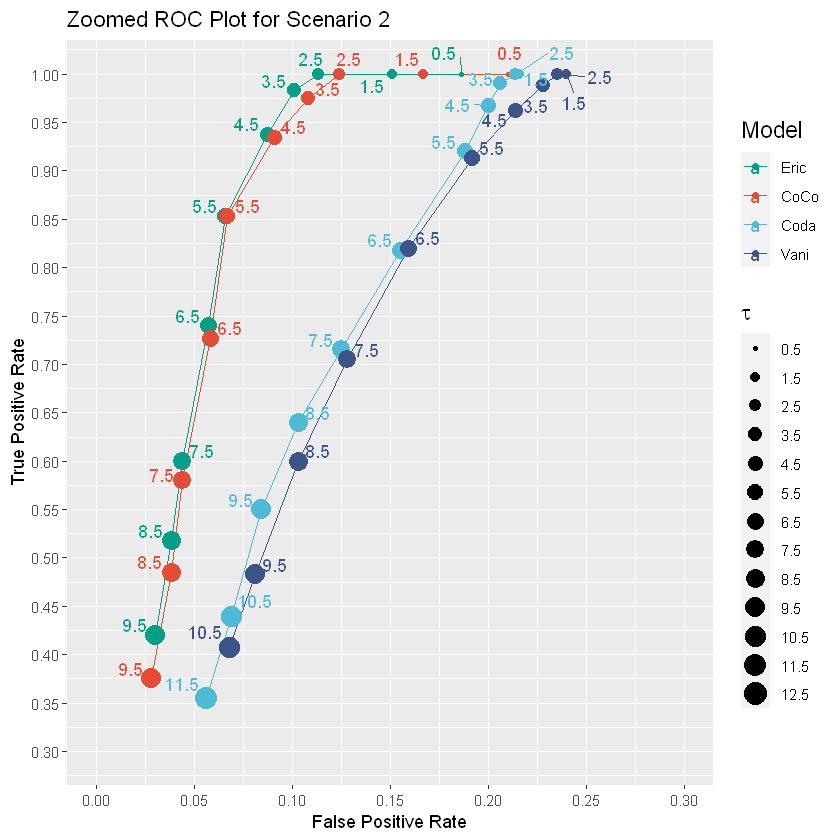}
\end{tabular}
\caption{ROC curves with different $\tau$ values under Scenario 2 of Simulation II. The left panel is ROC curves at original FPR and TPR scales and the right panel zooms in specific regions of FPR and TPR to better distinguish different methods.} 
\end{figure}

\end{document}